%% file: bhrd.tex
\def\draft{0}
\def\llncs{0}

\def\anon{0}

\ifnum\draft=0

\fi

\ifnum\llncs=0
\documentclass[11pt]{article}
\usepackage{fullpage}
\else
\documentclass[orivec,runningheads,envcountsame,envcountsect]{llncs}
\usepackage{amssymb,amsmath,cite,url}
\fi

\usepackage{bm,bbm}

\input{style}

\usepackage{qcircuit}

\title{Black-Hole Radiation Decoding is Quantum Cryptography%
\ifnum\llncs=1
\ifnum\anon=0
\thanks{For the most up-to-date version of this work, please refer to \url{https://arxiv.org/abs/2211.05491}.}
\fi
\fi
}

\ifnum\llncs=1
\author{Zvika Brakerski\thanks{%
		Supported by the Israel Science Foundation (Grant No.\ 3426/21), and by the European Union Horizon 2020 Research and Innovation Program via ERC Project REACT (Grant 756482).}
}
\institute{Weizmann Institute of Science, Israel}
\date{}
\else
\ifnum\anon=1
\author{}
\else
\author{Zvika Brakerski\thanks{Weizmann Institute of Science, Israel, \texttt{zvika.brakerski@weizmann.ac.il}. Supported by the Israel Science Foundation (Grant No.\ 3426/21), and by the European Union Horizon 2020 Research and Innovation Program via ERC Project REACT (Grant 756482).}
}
\fi

\ifnum\draft=1
\date{\today}
\else
\date{}
\fi

\fi

\ifnum\llncs=1
\renewcommand{\paragraph}{\boldpar}
\fi

\begin{document}

\maketitle

\begin{abstract}
We propose to study equivalence relations between phenomena in high-energy physics and the existence of standard cryptographic primitives, and show the first example where such an equivalence holds. A small number of prior works showed that high-energy phenomena \emph{can be explained} by cryptographic hardness. Examples include using the existence of one-way functions to explain the hardness of decoding black-hole Hawking radiation (Harlow and Hayden 2013, Aaronson 2016), and using pseudorandom quantum states to explain the hardness of computing AdS/CFT dictionary (Bouland, Fefferman and Vazirani, 2020).

In this work we show, for the former example of black-hole radiation decoding, that it also \emph{implies} the existence of secure quantum cryptography. In fact, we show an existential equivalence between the hardness of black-hole radiation decoding and a variety of cryptographic primitives, including bit-commitment schemes and oblivious transfer protocols (using quantum communication). This can be viewed (with proper disclaimers, as we discuss) as providing a physical justification for the existence of secure cryptography. We conjecture that such connections may be found in other high-energy physics phenomena.

\end{abstract}

\ifnum\llncs=0

\newpage
\tableofcontents

\newpage
\fi

\def\bhrd{\mathrm{BHRD}}
\def\efi{\mathrm{EFI}}
\def\wtW{\widetilde{W}}

\newcommand{\rdec}{\cA}
\newcommand{\sdec}{\cS}
\newcommand{\distn}{\cD}

\newcommand{\pieq}{\Pi_{\mathrm{EQ}}}

\section{Introduction}

The idea that computational considerations play a significant role in the most foundational laws of nature is gradually becoming well-accepted in physical research. Such ideas are recently coming up increasingly frequently in the context of understanding the plausible properties of the (still out-of-reach) theory of quantum gravity. 

One of the notable examples of this principle, which is the focus of this work, is due to Harlow and Hayden~\cite{HH13} (henceforth HH). They were considering the extent to which ``effective field theory'' (EFT, the combination of quantum field theory and the general theory of relativity which is applicable in ``mild'' gravitational conditions) provides accurate predictions about the physical universe. It is known that EFT is not applicable if gravity is very strong, but until recently it was widely conjectured to otherwise be universally applicable. Recent thought experiments involving black-holes are now challenging this conjecture and putting forth the question of the limits of applicability of EFT.

Whereas the singularity region of the black-hole (its ``center'') has extreme gravity and is not expected to be described by EFT, the event horizon of the black-hole (its ``boundary'') does not have particularly strong gravity. EFT therefore predicts that the near horizon area emits radiation known as Hawking radiation, which in turn is strongly believed to carry out information that fell into the black-hole throughout its history, in some scrambled form.%
\footnote{The reason that Hawking radiation is believed to carry information, is that otherwise the evolution of the wave-function of the universe would not be unitary, which seems implausible in the context of well accepted physical theory.} This means that, at some point, the entropy outside the black-hole must decrease. In other words, at some point, the outgoing Hawking radiation must be correlated with the previously-emitted (``historic'') radiation. State-of-the-art models for black-holes predict that this should happen almost immediately after the halfway point of the evaporation (this is known as ``Page time''), and that shortly after the Page time, outgoing radiation is almost fully quantumly entangled with the history. On the other hand, the general theory of relativity predicts that the near horizon area is not highly impacted by gravity (it is not very ``curved'') and therefore the near-horizon area contains ``ordinary'' vacuum that can be described by ``standard'' quantum field theory. The Hawking radiation therefore originates from the vacuum ground state in the horizon-area. According to quantum field theory, this ground state is strongly spatially entangled. Therefore, the outgoing Hawking radiation must be almost fully entangled with so-called ``modes'' in the near-horizon area. 
If a quantum system (the outgoing radiation) is (almost) fully entangled with another system (the near-horizon modes) then it cannot exhibit (almost) any correlation with another quantum system (the historic radiation). To resolve this difficulty, the ``complementarity principle'' was invoked. That is, it was suggested that if the outgoing radiation is maximally entangled with two different systems, then the two \emph{must actually be the same system}. Indeed, one of these systems (the historic radiation) lives ``outside'' the black-hole, and the other (the near-horizon modes) lives ``inside'' the black-hole. Therefore, perhaps the Hilbert space of the universe simply presents itself so that a part of it is accessible in two different ways, from inside and outside the black-hole. 

This is a very appealing solution, but it has been challenged by the so-called ``firewall paradox'' of Almheiri, Marolf, Polchinski and Sully~\cite{AMPS13}, who proposed the following thought experiment. 
An observer (Alice) can collect the historic radiation, and leisurely decode it. Alice can then verify that that the decoded value is indeed correlated with the outgoing radiation. Finally, Alice can ``jump in'' together with the radiation bit, and check that it is also correlated with the inside modes. Thus Alice witnesses a violation of the monogamy of entanglement, which is a fundamental feature of quantum theory.
The conclusion of \cite{AMPS13} is that the near-horizon area is not in a vacuum ground state, but rather in a very high-energy state that allows for the outgoing radiation to not be entangled with the near-horizon modes. If this is indeed the case, then EFT is inapplicable in the near-horizon area, despite it not having extreme gravitational conditions.

HH proposed to rescue complementarity using a computational argument. They argue that while Alice can in principle decode the historic radiation and distill a correlation with the outgoing radiation, the \emph{computational complexity} of this process may be prohibitive. In particular, the near-horizon modes will not survive long enough until the end of this calculation, so Alice can never witness a monogamy violation. It is therefore suggested that the ability to perform high-complexity tasks falls outside of the descriptive power of EFT, similarly to the setting of a strong gravitational field. See \cite[Section~6]{A16} for additional discussion of the firewall paradox from a computational perspective.

The HH decoding problem considers a quantum state defined over $3$-subsystems $HBR$, where $H$ represents the interior of the black-hole, $B$ is the outgoing radiation which for simplicity is assumed to contain a single bit, and $R$ is the historic radiation. The state is set up so that $BR$ are information-theoretically maximally correlated, that is there exists a \emph{possibly inefficient} quantum process that takes $R$ as input and ``distills'' it into a single qubit register $D$ that is maximally entangled with $B$ (the state of $BD$ is close to $\epr = \frac{\ket{00}+\ket{11}}{\sqrt{2}}$). The HH decoding problem is to efficiently implement such a procedure that takes $R$ and produce $D$ that is highly entangled with $B$.\footnote{The description of the quantum operations in \cite{HH13} and in followups \cite{S13,A16} is in terms of unitaries that are efficient/inefficient to implement. We find this unnecessarily cumbersome and consider general quantum procedures, which correspond to completely positive trace preserving (CPTP) channels. Such channels can always be ``purified'' into unitary form, so this does not limit the generality of the discussion.} They show that this problem is potentially intractable \emph{in the worst case} (i.e.\ there is no efficient process that solves \emph{every possible} $HBR$), unless an implausible complexity-theoretic result follows. Specifically, that all languages that have a statstical zero-knowledge protocol also have a quantum polynomial-time algorithm ($\szk \subseteq \bqp$).

Aaronson~\cite{A16} showed a similar result, but under a different computational assumption. He relied on the existence of injective one-way functions (injOWF) that are secure against quantum inverters. In a nutshell, injOWF is an injective function $f$ such that $f(x)$ can be computed in time $\poly(|x|)$, but $f$ cannot be inverted in polynomial time. Aaronson suggests that this subsumes the result of HH, since $\szk$-hardness implies the existence of OWF \cite{O91}. We believe that this claim is true in spirit but formally it is not completely accurate.\footnote{This is not the focus of this work, so we choose not to elaborate too much on this point. In a nutshell, there are two gaps in the claim of \cite{A16}. One is that \cite{O91} refers to \emph{average-case hardness} in $\szk$, whereas the HH argument only assumes \emph{worst-case hardness}. To an extent, HH rely on a weaker assumption but also prove a weaker claim: they only show that there does not exist an efficient decoder that decodes \emph{all} possible black-holes, whereas Aaronson shows that there is a (single) potential black-hole that resists all efficient decoders. The second point is that \cite{O91} shows that average-case $\szk$ hardness implies one-wayness, but not necessarily injectivity.}

Aaronson further suggests that the existence of injOWF implies a stronger sense of decoding-intractability, as follows. HH and Aaronson's first result presented a system in which $BD$ cannot be brought to the EPR state, but it was still trivial to infer classical correlations between $B$ and $R$. Namely to produce a $D$ that is classically correlated with $B$. In Aaronson's second construction, classical correlations are shown to be intractable to detect. This is done using Goldreich and Levin's hard-core-predicate construction \cite{GL89}. We point out that, whereas finding classical correlation becomes intractable in this latter construction, it becomes trivial to find \emph{phase correlations}, namely, correlations between $BD$ when both registers are presented in the Hadamard basis. In this sense, Aaronson's second construction does not actually give a qualitative improvement over the first one. Indeed, it appears that prior works were not very concerned with the possibility that some correlation between $B$ and $R$ may be efficiency detectable, so long as recovering perfect quantum correlation was intractable. %

\subsection{This Work: Gravitational Cryptology}

The arguments of \cite{HH13,A16} show that under computational/cryptographic assumptions, it is possible to resolve an quantum-gravity ``paradox''. That is, in our terminology, a cryptology-induced gravitational phenomenon. In this work, we improve this connection, by showing that \emph{even milder} cryptographic constructs, namely ones that are not known to imply OWF at all (whether injective or not) already give rise to potential ``black-hole radiation'' whose decoding is intractable. We show that these extremely mild assumptions also allow us to construct ``radiation'' which is not only intractable to decode, but one that is \emph{computationally indistinguishable from being completely uncorrelated} with the historic radiation $R$. Therefore, it is impossible for any efficient process to extract \emph{any} correlation between $B$ and $R$.

We then turn our attention to show the complementary result. Namely, to show that if there exists a black-hole with hard-to-decode radiation, then such a physical phenomena would imply the existence of cryptographic objects! We believe that such a result has not been previously shown in the quantum context. To some extent, one can view this as providing some physical justification to the existence of cryptography.\footnote{At the same time, we caution from over-interpreting our result. The hardness of radiation decoding has been suggested as a possible solution to an apparent paradox, under a specific set of assumptions about the behavior of quantum gravity and black-holes. Therefore our result is only meaningful in the context of these physical assumptions. Nevertheless, the argument that firewalls should be avoided rests on physical understanding of the universe, and being able to derive computational conclusions from it appears worthwhile.
See additional discussion in Section~\ref{sec:intro:discussion} below.
}

Putting the two results together, we get an \emph{existential equivalence} between the hardness of decoding Hawking radiation, using a proper asymptotic formalism that we propose (see Technical Overview below for additional details), and the ability to securely realize a variety of cryptographic objects. This variety includes: bit-commitment schemes (which allow a party to fix a binary value without revealing it to another party, but such that the other party can later verify what that fixed value was), oblivious transfer protocols (a protocol where in the beginning one party holds two bits $a_0, a_1$, and the other holds a bit $b$, and in the end the latter party holds $a_b$, without learning $a_{\bar{b}}$ and without the first party learning $b$ itself), secure multiparty computation protocols (where parties can jointly and securely compute arbitrary functions on their inputs), and (non-trivial) quantum zero-knowledge proof protocols for languages in the complexity class $\qip$. In fact, we rely on a sequence of works, summarized in \cite{BCQ22}, that shows that all of the above cryptographic primitives are existentially equivalent to a very basic object that is called $\efi$ in \cite{BCQ22}. $\efi$ is simply a pair of distributions (over quantum states) that are statistically distinguishable but computationally indistinguishable.\footnote{$\efi$ stands for Efficient (to sample), Far (statistically), and Indistinguishable (computationally). For a formal definition, see Definition~\ref{def:efi}.} We show that this cryptographic object corresponds to hard-to-decode Hawking radiation.

Having shown an equivalence between finding a non-trivial correlation between $B$ and $R$ is equivalent to the existence of a cryptographic primitive, we go back to consider the implications on the \cite{AMPS13} experiment. We notice that in the experiment, Alice processes $R$, compares it with $B$ and then jumps into the black-hole. Whereas any non-trivial correlation between $BR$ would suggest that something strange is going on, Alice's level of confidence may be quite low, and in particular may not be worth jumping into the black-hole to find out the answer. It would be quite unsatisfactory if cryptology is only implied by extremely-hard-to-decode ``black-holes''.\footnote{Alice can boost her confidence by performing the experiment many times and relying on concentration bounds. This would require her to chase after multiple radiation particles once inside the black-hole. This is not unreasonable but we view it as less elegant compared to our solution.} We therefore show that even ``mildly hard to decode'' black-holes imply $\efi$. Therefore, by our prior result, that mildly-hard-to-decode black-holes imply potential extremely-hard-to-decode black-holes. Our notion of mild-hardness is one where Alice attempts to decode $R$, and declares whether she succeeded or not. A \emph{high-confidence decoder} is one that declares success reasonably often (i.e.\ with non-negligible probability), and, conditioned on declaring success, it produces a distilled qubit $D$ such that $BD$ are extremely correlated (the fidelity with an EPR pair is very close to $1$).

\subsection{Technical Overview}

In order to show anything about the Black-Hole Radiation Decoding problem (which we denote by $\bhrd$), we must provide a formal definition. We provide a somewhat more quantitative description compared to the one appearing in prior works \cite{HH13,S13,A16,A22}. Similarly to those prior works, we consider an efficiently generatable quantum state that represents the state of the black-hole interior $H$, the historic radiation $R$ and the next outgoing bit of radiation $B$. We are guaranteed that $BR$ are very strongly correlated, in the sense that there exists an \emph{inefficient} quantum procedure that takes $R$ as input and produces single qubit $D$ such that $BD$ form an EPR pair.\footnote{As noted above, our formulation is in terms of quantum channels and not in terms of unitaries.} In our formulation we allow a little slackness so that $BD$ are allowed to only be very close to an EPR pair (in the sense that their fidelity with the EPR state is negligibly close to $1$). The decoding problem is to efficiently ``distill'' this entanglement. 

But how much entanglement should be hard to distill? In prior works \cite{HH13,A16}, the intractability was shown for getting fidelity $1$ (or close to $1$), and the counterexamples provided in fact exhibited weak non-trivial correlation. We notice that if we wish to completely avoid the firewall paradox, we need to rule out $BD$ having any sort of non-trivial correlation. Therefore, we consider a \emph{radiation decoder} $\rdec$ to be successful if it manages to take $R$ as input, and produce $D$ such that the projection of $BD$ on the EPR state (or rather, the square of its absolute value) is non-negligibly greater than $1/4$. Note that $1/4$ correlation can easily be achieved by outputting $D$ in a maximally mixed state (that is, completely uncorrelated with $B$). However, as we mentioned above, we show an amplification theorem that shows that if there exists a state for which ``high confidence'' decoding is hard then there is also one for which any non-trivial decoding is hard.\footnote{We recall that high-confidence decoding means that with some non-negligible probability, the decoder declares success and produces $D$ such that the fidelity of $BD$ with EPR is $1-o(1)$.}

\paragraph{Cryptology-Induced Gravitational Phenomena.} 
We show that $\efi$ imply a hard-to-decode state as follows. Our hard-to-decode state will start from an EPR pair in registers $BT$. It will then apply $X^x Z^z$ on the $T$ part of the EPR, and use \emph{a commitment scheme} to produce commitments to both $x$ and $z$. The register $R$ will constitute of $T$ and the commitments to $x,z$. We use a cryptographic primitive known as a \emph{non-interactive statistically-binding bit-commitment} which is essentially equivalent to $\efi$ (see e.g.\ \cite{Yan22}). Intuitively, for the purpose of this work, a commitment string is similar to an encryption, in the sense that it computationally hides a well-defined committed value. We can thus think of the radiation decoding task as the task of unmasking $T$. This is information-theoretically possible, since a commitment to a value $x$ allows us to recover $x$ using an inefficient procedure. However, to a computationally bounded adversary, a commitment to $x$ is indistinguishable from a commitment to a fixed value (say $0$). Therefore, a computationally bounded decoder cannot deduce any information about $x,z$ from the commitments, and its view is therefore equivalent to seeing $T$ masked with an unknown random Pauli. This is equivalent to $T$ being maximally mixed, and in particular independent of (i.e.\ in tensor product with) $B$. Therefore a computationally bounded decoder cannot detect \emph{any} correlation between $B$ and $R$ without violating the properties of the commitment and thus of the $\efi$.

We note that, as explained above, in our construction, $B$ and $R$ are computationally indistinguishable from being in a tensor product, and therefore it is impossible to find any correlation between them. Therefore, our result is stronger than prior ones by \cite{HH13,A16}. Indeed, in hindsight, one can view the previous results of \cite{HH13,A16} as having a similar structure, but only masking either in the $X$ or $Z$ ``basis'' but not in both, thus not achieving the strong computational tensor-product property that we achieve, and using stronger computational assumptions.

\paragraph{Superdense Decoding, Gravitationally-Induced Cryptology.} To show the converse result, that $\bhrd$-hardness implies $\efi$, we show an equivalence between the $\bhrd$ problem and another problem that we call ``superdense decoding''. Essentially we observe that distilling $D$ that is jointly EPR with $B$ is \emph{quantitatively} the same task as being able to decode two classical bits that are superdense-encoded in $B$. A more detailed description follows.

We recall that when two parties share an EPR pair on registers $BD$, it is possible to transfer two classical bits $x,z$ from one to the other using a single qubit of communication. This is done by applying a Pauli operation $X^x Z^z$ on $B$ and then sending $B$ to the holder of $D$. One can verify that the $x,z$ values generate a set of orthogonal states (the Bell/EPR basis) on $BD$, and thus $x,z$ can be recovered. We present a formal argument that the task of decoding $D$ from some register $R$ is equivalent to the task of decoding a superdense encoding of $x,z$ over $B$. That is, a superdense decoder $\sdec$ takes two registers $BR$, where $B$ encodes random values $x,z$ using the superdense encoding procedure, and its goal is to produce the classical values $x,z$. The \emph{advantage} of the superdense decoder is the probability that it recovered both $x,z$ correctly.

Showing that a radiation decoder implies a superdense decoder (on the same initial state $HBR$) is fairly straightforward. If $D$ that is in EPR with the original $B$ can be decoded from $R$, then we can apply the regular superdense decoding to recover $x,z$. For the converse, we rely on quantum teleportation. 
Assume we have a ``good'' superdense decoder $\sdec$, and we wish to derive from it a good radiation decoder $\rdec$. The superdense decoder produces a classical output $x,z$ out of a pair of registers $BR$, whereas the radiation decoder is required to produce a quantum register $D$ out of the radiation $R$. To do this, we give $\sdec$ as input the $R$ for which we want to decode the radiation, but we give it a ``fake'' $B$ which is a half-EPR of a freshly generated EPR pair.
We show that a successful superdense decoding corresponds to \emph{teleporting} an implicit qubit $D$ onto the other end of the EPR pair of ``fake $B$''. Therefore, even though the superdense decoder only returns a classical value, we can use this value to obtain the decoded radiation qubit $D$. 

The correspondence between the two problems is \emph{exact}. Namely, the probability of recovering $x,z$ by a superdense decoder is identical to the fidelity of the radiation decoder output with an EPR pair, and this is achieved with roughly the same complexity (in both directions). Thus, the superdense decoding problem provides a convenient proxy for the radiation decoding problem, only now we have a problem with a classical interface (random classical bits $x,z$ go in, get encoded quantumly, sent to a quantum decoder, but the output of this decoder is again a classical value that is compared against $x,z$).

To generate $\efi$ from a superdense-hard state, we consider an instance of the superdense decoding problem in registers $BR$, together with the ``correct answer'' $x,z$. We define two distributions over quantum states. One that outputs $BRxz$ and one that outputs $BRu_1u_2$, where $u_1, u_2$ are random bits that are drawn independently of $x,z$. The two distributions are clearly distinguishable by an inefficient procedure that superdense-decodes $BR$ and compares the outcome with the two classical bits provided. We show that distinguishing the two with any non-negligible probability allows us to ``learn'' about the values of the actual $x,z$ and thus to obtain non-negligible advantage against the superdense decoding problem.

\paragraph{Amplification.} 
We are given a state which is only hard to superdense-decode with high confidence. Namely, it is hard to come up with a decoder that declares ``success'' with non-negligible probability, and upon declaring success, correctly decodes the state with probability $1-o(1)$. We would like to show that even such a hard-to-decode state implies $\efi$. The argument above is clearly insufficient, since it is possible that the two distributions obtained are always distinguishable to within a non-negligible advantage.

Our amplification result is achieved using repetition. We first show that if we get a sequence of decoding instances, and are tasked with correctly decoding \emph{all of them}, then we cannot succeed with any non-negligible probability. We then use a hard-core bit construction to show a pair of $\efi$ distributions.

For the first part, we consider a challenge that consists of an $n$-tuple of independent decoding instances that we call ``slots''. We consider the task of correctly decoding all slots with non-negligible probability. 
We show that given such a decoder, that succeeds with non-negligible probability $\epsilon$, it is possible to identify a slot $i$ and an event $E_i$, such that $E_i$ occurs with probability at least $\epsilon$, and when $E_i$ occurs, the probability that the decoding of the $i$-th slot succeeds is roughly $1-1/n$. The proof of this ``Estimation Lemma'' uses elementary probability theoretic techniques. Specifically, $E_i$ is simply the event that all slots $<i$ were decoded correctly. We show that at least for some of the $i$'s, if all previous slots were correct, then the $i$-th slot will also decode correctly.

This establishes the hardness of of decoding all $n$ slots with non-negligible probability. The $\efi$ construction follows by considering an instance that contains $n$ slots as before, but also the inner product of the concatenation of all correct answers with a random vector $a$, in addition to $a$ itself. This is exactly the Goldreich-Levin hard core bit \cite{GL89}. We set one of the $\efi$ distributions to have the correct inner product value, and the other to have its negation. This means that a distinguisher between the two amounts to calculating the value of the inner product.
We then use a result by Adcock and Cleve \cite{AdcockC02} that shows a quantum algorithm that converts such an inner-product predictor into one that produces the entire pre-image of the inner product. This completes the amplification algorithm.

\subsection{Discussion and Other Related Work}
\label{sec:intro:discussion}

Our work suggests that, to some extent, the existence of cryptographic hardness could rely on a natural phenomena. The desire to find physical substantiation for cryptography is not new. Indeed, one can make a (classical) thermodynamical argument that the increase of entropy in complex systems implies that some processes (like scrambling an egg or breaking a glass) are ``efficiently occurring but intractable to invert''. The argument that is implied by this work follows a similar high level logic, but the quantitative nature of the black-hole radiation decoding problem allows us to draw quantitative conclusions and explicitly point out to the connection between the (proposed) properties of physical theory, and computational hardness.

It is nevertheless important to state as clearly as possible what we need to assume in order to derive the conclusion that cryptographic hardness is ``implied'' by nature.
\begin{enumerate}
	\item We assume that the description of black-hole dynamics is indeed as accepted physical theory suggests. This includes the unitarity of the evolution of black-holes and the applicability of ``known physics'' (effective field theory) in the near-horizon area.
	
	\item We assume that there is some cosmological censorship of the firewall paradox that prevents the \cite{AMPS13} experiment from being successfully carried out ``under normal circumstances''. Namely, given that we live in ``moderate gravity space'' and haven't encountered random firewalls popping into existence, we should not expect any experiment we conduct to bring them up in the specific moderate gravity space in the near-horizon area.
	
	\item We accept the translation of the decoding problem from a physical system into an information theoretic task of decoding a qubit of information. Or at least that a more accurate description does not significantly change the nature of the problem.
	
	\item We state the decoding problem as an average case problem, namely we require a black-hole which is hard to decode against all efficient decoders. We furthermore assume that this hardness needs to be super-polynomial. It may be the case that some concrete hardness would be required, rather than just assuming hardness against all polynomial time decoders. For example, \cite[Footnote 4]{HH13} suggests that exponential hardness of the decoding problem should be sought out. We note that our reductions are quite generic and tight, and asymptitcally preserve the hardness of the decoding problem versus that of the obtained cryptographic primitive. They can therefore be applied to any ``degree of hardness'' of the decoding problem.
\end{enumerate}

Whereas each of these axioms can be challenged, we believe that they are at least consistent with the way science currently views the physical world. Therefore, while we cannot formally deduce cryptographic hardness from nature, we can at least view the above as some justification that cryptography is a natural part of our current view of the universe.

\paragraph{Quantum Cryptography vs.\ Classical Cryptography.} In this work we show that the black-hole radiation decoding problem is equivalent to the existence of quantum-cryptographic objects, namely EFI. However, it remains unsettled whether it has implications on the landscape of \emph{classical} cryptography. Aaronson \cite{A16} raised the question whether radiation decoding remains intractable even in the extreme case where $\ccp$=$\pspace$ (where in particular classical cryptography does not exist). Our work frames this question in the context of the relation between quantum cryptography and classical complexity. It was recently shown by Kretschmer~\cite{Kretschmer21} that (in a quantum-relativized world) quantum cryptography, and in particular EFI, could exist even if $\bqp$=$\qma$ (which is a setting where in particular classical cryptography can always be efficiently quantumly broken). The same work also showed that some forms of quantum cryptography cannot exist if $\ccp$=$\pspace$, but this refers to so-called pseudorandom quantum states which are plausibly stronger primitives than EFI. It remains an intriguing open problem whether EFI could exist even if $\ccp$=$\pspace$, and this work shows the implications to the radiation-decoding problem.

\paragraph{Gravitational Non-Malleability, Black-Hole Dynamics.} 
Kim, Tang and Preskill~\cite{KTP20} were considering the question of whether complementarity between the emitted radiation and the black-hole interior implies that the outside system could effect the black-hole interior by applying operations on the system $R$. They used cryptography, and in particular pseudorandom quantum states \cite{JLS18} in order to show that under the assumption that the emitted radiation is pseudorandom (and thus ``sufficiently scrambled''), it is not possible to act on it \emph{efficiently} in a way that would effect the black-hole interior in a non-negligible manner. This is an additional evidence that the connection between cryptography and gravitational phenomena may not be incidental. Nevertheless, it is not known whether black-hole dynamics would actually emit the radiation in a pseudorandom manner. Indeed, coming up with a satisfactory description of information processing in black-hole dynamics is a grand goal, and our work is another evidence that cryptographic research may play a significant role in pursuing it.

\paragraph{Gravitational Cryptology in AdS/CFT?}
Cryptographic tools have recently been used by Bouland, Fefferman and Vazirani \cite{BFV20} in order to characterize computational properties of the so-called AdS/CFT correspondence \cite{Maldacena98}. At a high level, AdS/CFT is a class of quantum gravity theories that are based on a duality between ``quantum gravity in the universe'' (AdS, this is often called ``the bulk'') and ``quantum fields on the boundary of the universe'' (CFT).\footnote{To be precise, our universe is not AdS, so AdS/CFT theories do not directly apply to universes like our own. Nevertheless, it is believed that AdS/CFT hints to a duality relation that could exist in our universe as well.} It is shown in \cite{BFV20} that quantum gravity theories such as AdS/CFT are either hard to use in the sense that the AdS to CFT translation is computationally intractable, or that they can solve computational problems that are intractable to ``standard'' quantum computers. Their work relied on the conjectured security of an explicit construction of a cryptographic object (a pseudorandom quantum state \cite{JLS18}). Gheorghiu and Hoban \cite{GH20} later proposed a path towards establishing a similar argument based on a well studied cryptographic problem (Learning with Errors \cite{Regev05}).

The idea that the AdS/CFT mapping is potentially cryptographically hard, at least in some instances, could lead to a speculation that gravitational cryptology could be emergent in that context as well. Is it possible that the properties of quantum gravity \emph{necessitate} cryptographic hardness in the AdS/CFT map? Can we think of the CFT description as an encryption of the AdS space? Interestingly, Aaronson (attributing to a discussion with Gottesman and Harlow about a proposal by Susskind) \cite{AGHS22} proposed that such cryptographic hardness would posses \emph{homomorphic} properties. This is since the evolution of the system in the bulk, according to the laws of physics, is represented by the evolution in the CFT, but the ``plaintext'' is intractable to recover. Many parts of this analogy remain unclear and await further investigation.

\section{Preliminaries}

\paragraph{The Security Parameter.} The notion of a security parameter is one of the most fundamental in cryptography and is used to provide an asymptotic notion of efficiency and security, especially in cases where the input lengths do not immediately translate to a complexity measure (e.g.\ algorithms that do not take an input).

All algorithms we consider take an additional implicit input called the \emph{security parameter} and usually denoted by $\secp$. The complexity of the algorithm is defined as a function of $\secp$. Therefore, in fact, every algorithm we consider constitutes of an ensemble of algorithms, one for each value of $\secp$. Efficient algorithms are ones that run in time $\poly(\secp)$ for some polynomial, and a negligible function is one that vanishes faster than any inverse polynomial function in $\secp$.

\paragraph{Quantum States and Registers.} We use standard notation for quantum states and quantum registers. Quantum states are vectors in complex Hilbert space. For the purpose of this work, we consider Hilbert spaces composed of qubits, namely of the form $\bigotimes^n_{i=1} \bbC^2 \cong \bbC^{2^n}$, where $n$ is the number of qubits. We associate Hilbert spaces with ``registers'' that manifest the local property of individual qubits. Registers are denoted by capital letters such as $A,B, \ldots$, and the concatenation of registers corresponds to the tensor product of their respective Hilbert spaces. Every $n$-qubit Hilbert space can naturally be written as a sequence of $n$ single qubit registers $A_1 \cdots A_n$. We can thus refer to the register $A_1 A_3$, which corresponds to the first and third qubits in the sequence. We may sometimes use more than one partitioning for the same Hilbert space, for example if $n=5$ we may consider a $2$-qubit register $B$ and an $3$-qubit register $C$, and define $B = A_1 A_3$, $C = A_2 A_4 A_5$. In this case $BC$ representes a Hilbert space that is isomorphic to $A_1 \cdots A_5$. Likewise $CB$ is isomorphic to $BC$.

We use standard Dirac notation for quantum states. We recall that quantum states are fully characterized by their density matrix and that ``pure quantum states'' are ones represented by a rank-$1$ density matrix. The density matrix of a quantum state that is supported over a number of registers, say $A$,$B$ is denoted $\rho_{AB}$, and the reduced density matrix over register $A$ is denoted $\rho_A$. 
A register is said to contain a classical value if its reduced density matrix is diagonal in the computational basis (see below). We often refer to classical values by denoting them in lowercase letters.

A pure state $\ketpsi$ over $m$ qubits is a \emph{purification} of a (possibly mixed) state $\rho$ over $n \le m$ qubits, if letting $A,B$ be registers of length $n, (m-n)$ respectively, and setting $\rho_{AB}=\ketbra{\psi}$, it holds that $\rho_A = \rho$. Every state has a purification of length $m=2n$.

Every qubit Hilbert space is associated with a basis that we refer to as the ``computational basis'' and denoted $\{\ketz, \keto\}$. We also refer to the ``Hadamard basis'' $\{\ketp, \ketm\}$ which is defined respective to the computational basis by applying the Hadamard unitary operator.

We use standard notation for quantum gates (unitary operators) such as Pauli operators $X,Z$, Hadamard $H$ and controlled $X$/$Z$ operations: $\cnot: \ket{xy} \to \ket{x,x\oplus y}$, $\cz: \ket{xy} \to (-1)^{xy}\ket{x,y}$. When we wish to apply a gate or unitary on a specific register, we denote it in subscript. For example, in a system containing registers $ABCD$, if $B$ is a single qubit register then $H_B$ denotes the operation of applying the Hadamard gate on the register $B$ (and acting as identity on $ACD$).

We let $\epr$ denote the standard two-qubit EPR state $\epr = \frac{\ket{00} + \ket{11}}{\sqrt{2}}$, and also consider the ``EPR/Bell basis'' $\{ \ket{\eprtxt_{xz}} = (I \otimes X^x Z^z) \epr \}_{x,z\in \binset}$. Naturally, $\epr = \ket{\eprtxt_{00}}$. We note that $(H \otimes I) \cnot \ket{\eprtxt_{xz}} = \ket{xz}$.

\paragraph{Quantum Channels.} A quantum channel is a CPTP (completely positive trace preserving) map between density operators, where the input and output can be defined over different Hilbert spaces. We refer to the Hilbert space of the input as the ``source'' and that of the output as the ``target''. Given two registers $S,T$ whose associated Hilbert spaces correspond to those of the source and target spaces of a channel $\cC$, we let $\cC[S\to T]$ denote the application of the channel on the quantum state stored in register $S$, and storing the output in the register $T$. %
We think of the application of a channel as destroying the source register $S$ and generating the target register $T$.

The computational complexity of a quantum channel refers to the minimal quantum circuit size that implements it. The atomic operations in quantum circuits are the application of a unitary on a constant number of qubits (say at most $3$ qubits), introducing additional single-qubit registers that are initialized to the zero state, and discarding single-qubit registers (thus tracing them out of the total state of the system). It follows from the above that generating a random classical bit, and applying constant-arity classical gates can also be performed by a quantum circuit at a constant cost.

An ensemble (indexed by the security parameter) of quantum channels is efficiently implementable if its complexity is bounded by $p(\secp)$ for some polynomial $p$. An ensemble of quantum states is efficiently generatable if there exists an efficient channel that generates them (starting from an empty, or $0$-dimensional, register).

For any quantum channel $\cC$, there exists a unitary $U_{\cC}$ such that $\cC[S\to T]$ is equivalent to the following. Consider registers $A, A'$ such that the dimensions of $SA$ and $TA'$ match, apply the unitary $U_{\cC}$ on $SA=A'T$ and trace out $A'$. Note that this unitary is not unique (although it is possible to define such a unitary canonically). The channel $\cC$ is efficiently implementable if there exists such a $U_{\cC}$ which is efficiently implementable. The unitary $U_{\cC}$ is a \emph{purification} of $\cC$. If a channel (efficiently) generates some quantum state, a purification of the channel (efficiently) generates a purification of the state.

In terms of notation, we may refer to ``a channel from register $A$ to register $B$'' when we mean that the channel is defined over the Hilbert spaces induced by these registers. Such a channel can be applied on other registers that are associated with Hilbert spaces of the same dimension as $A,B$.

\paragraph{Measurements and Distinguishers.} For the purpose of this work, we use the following definition of a measurement. A quantum measurement $\cT$ of a state in a register $S$ is a quantum channel from $S$ to a classical register $t$. A computational basis measurement of an $n$-qubit register corresponds to applying the operator $\sum_{x \in \binset^n} \ketbra{x} \rho_S \ketbra{x}$. More generally measurement in a basis represented by a unitary $U$ corresponds to the operator $\sum_{x \in \binset^n} U \ketbra{x} U^\dagger \rho_S U \ketbra{x} U^\dagger$. Any projector $\Pi$ on the Hilbert space associated with a register $S$ induces a binary (two-outcome) measurement on $S$ by $\ketbra{1}\tr\left[{\Pi \rho_S}\right] + \ketbra{0}\tr\left[{(I-\Pi) \rho_S}\right]$. That is, $\Pr[t=1]=\tr\left[{\Pi \rho_S}\right]$. We refer to such a measurement as ``applying the projector $\Pi$'', and refer to the outcome $1$ as an ``accepting'' event for the measurement. We denote $\pieq = \ketbra{00} + \ketbra{11}$.

A distinguisher between two distributions is a quantum channel whose target is a single classical bit. Such a channel $\cD[S \to t]$ can always be represented by an operator $0 \le {\cD} \le I$ over $S$ (note the slight overloading of notation for both the channel and the operator), s.t.\ $\tr[{\cD} \rho_S]$ is the probability that the channel outputs the value $1$ (when $\rho_S$ is the state of $S$). The \emph{advantage} of a distinguisher in distinguishing two distributions $\rho_0, \rho_1$ is $\tr[{\cD} (\rho_1-\rho_0)] \in [-1, 1]$.\footnote{We use a \emph{signed} definition of advantage, where positive advantage reflects positive correlation between the output of the distinguisher and the ``label'' of the input state, and negative advantage represents negative correlation.} The maximal distinguishing advantage between two distributions is called their \emph{statistical distance} (or total variation distance) and is equal to $\norm{\rho_1 - \rho_0}_1$.

Consider an ensemble of pairs of distributions (indexed by the security parameter $\secp$). We say that two distributions are \emph{statistically indistinguishable} if their statistical distance as a function of $\secp$ is a negligible function. We say that two distributions are \emph{computationally indistinguishable} if every \emph{efficient} distinguisher (a sequence of distinguishers with complexity polynomial in $\secp$), has negligible advantage.

The following propositions are straightforward.
\begin{proposition}\label{prop:indisthybrid}
	Consider two distributions $\rho_0, \rho_1$, and define $\rho'_0 = \tfrac{1}{2} I \otimes \rho_0$ and $\rho'_1 = \tfrac{1}{2} \sum_b \ketbra{b} \otimes \rho_b$. Then if there exists $\cD'$ with advantage $\epsilon$ against $\rho'_0, \rho'_1$, then there exists $\cD$ with advantage $\epsilon/2$ against $\rho_0, \rho_1$ with essentially the same complexity. %
\end{proposition}
\begin{proof}
	Given $\cD'$ as in the statement, define $\cD$ as follows. Given an input in a register $S$, initialize a register $B$ to the state $\ketbra{1}$ and run $\cD'$ on $BS$. The result follows by calculation.
\end{proof}

\begin{proposition}\label{prop:distinguish}
	Consider two distributions $\rho_0, \rho_1$ and let $\cD$ be a distinguisher with advantage $\epsilon$. Then the following experiment succeeds with probability $1/2+\epsilon/2$.
	\begin{enumerate}
		\item Initialize registers $BS$ in state $\tfrac{1}{2} \sum_b \ketbra{b} \otimes \rho_b$.
		
		\item Apply $\cD[S \to D]$.
		
		\item Apply $\pieq$ on $BD$ (i.e.\ output the negation of their XOR in the computational basis).
	\end{enumerate}
\end{proposition}
\begin{proof}
	Let $O$ be the operator associated with $\cD$. Define $\delta_b = \tr[O \rho_b]$. Then the state of $BD$ after the application of $\cD$ is 
	\begin{align}
		\tfrac{1}{2}\left( (1-\delta_0) \ketbra{00} + \delta_0 \ketbra{01} + \delta_1 \ketbra{11} + (1-\delta_1)\ketbra{10} \right)~.
	\end{align}
	The success probability is thus $\tfrac{1}{2} (1-\delta_0 +\delta_1) = 1/2+\epsilon/2$.
\end{proof}

\paragraph{Quantum Teleportation and Superdense Coding.}
One of the most basic features of quantum information is the ability to teleport a quantum state using classical communication. Let $TDBB'$ be a quadruple of quantum registers, where $T,D,B$ are all single-qubit registers, and $B'$ is of an arbitrary dimension. Consider an initial state where $\rho_{BB'} = \rho_0$ is arbitrary and $\rho_{TD} = \kbepr$. Then if we measure $DB$ in the Bell basis, namely apply $(H_B \otimes I_D) \cnot_{BD}$ and measure in the computational basis, and obtain an outcome $x,z$, then it holds that post-measurement: $\rho_{TB'} = ((X^x Z^z)_B \otimes I_{B'}) (\rho_0)_{BB'}$. Namely, given the values of $x,z$ it is possible to recover $T$ to have the exact same state as the original $B$.

A related phenomenon is that of superdense coding, which allows us to send two bits of classical information using just one qubit, assuming that the parties pre-share an EPR pair. Let $BD$ be a register that is initialized to $\epr$. Now to encode two classical bits $x,z$, apply $X^{x}Z^{z}$ on the register $D$, and send it to the holder of $B$. Now once $BD$ are measured in the Bell basis, $x,z$ are recovered.

\paragraph{The Quantum Goldreich-Levin Theorem.} The following is a quantum version of the famous Goldreich-Levin theorem \cite{GL89}. This version was proven by Adcock and Cleve in \cite{AdcockC02}, although it is not stated as a standalone theorem but rather as Eq.~(16) which describes an algorithm presented in Figure~1 and is proven as the first part of the proof of {\cite[Theorem~2]{AdcockC02}}.

\begin{theorem}[{\cite[Eq.~(16)]{AdcockC02}}]\label{thm:qgl}
	Let $n \in \bbN$, $x \in \binset^n$ and let $\{U_r\}_{\binset^n}$ be a set of $m$-qubit unitaries. Define the unitary $U$ that applies $U_r$ controlled by a value $r$ received in an additional register, formally: $U = \sum_{r} \ketbra{r} \otimes U_r$. Then the quantum circuit $\cR^{U, U^\dagger}$ described in Figure~\ref{fig:ACfig1} has the following property.

	If, on average over $r$, the first qubit of $U_r \ket{0^m}$ is $\epsilon$-correlated with $r\cdot x = \sum_{i=1}^n r_i x_i \pmod{2}$, namely
	\begin{align}\label{eq:ippredict}
		\Ex_r \left[ \norm{ (I \otimes \bra{r \cdot x}) U_r\ket{0^m}}^2 \right]  \ge 1/2 + \epsilon~,
	\end{align}
	then it holds that 
	\begin{align}
		\abs{\bra{x, 0^m, 1} \cR^{U, U^\dagger}\ket{0^{n+m+1}}} \ge 2\epsilon~.
	\end{align}
	
\end{theorem}

\begin{figure}[h]
	\centering
	\setlength{\unitlength}{0.5mm}
	
	\begin{picture}(50,120)(0,0)
	\end{picture}
	\begin{picture}(165,120)(17,0)
		
		\put(5,88){$\left\{ \mbox{\rule{0mm}{9mm}} \right.$}
		\put(5,43){$\left\{ \mbox{\rule{0mm}{9mm}} \right.$}
		
		\put(-17,88){\makebox(10,4){$n$ qubits}}
		\put(-17,43){\makebox(10,4){$m$ qubits}}
		
		\put(15,15){\line(1,0){10}}
		\put(35,15){\line(1,0){50}}
		\put(95,15){\line(1,0){70}}
		
		\put(15,30){\line(1,0){25}}
		\put(80,30){\line(1,0){20}}
		\put(140,30){\line(1,0){25}}
		
		\put(15,45){\line(1,0){25}}
		\put(80,45){\line(1,0){20}}
		\put(140,45){\line(1,0){25}}
		
		\put(15,60){\line(1,0){25}}
		\put(80,60){\line(1,0){20}}
		\put(140,60){\line(1,0){25}}
		
		\put(15,75){\line(1,0){10}}
		\put(35,75){\line(1,0){5}}
		\put(80,75){\line(1,0){20}}
		\put(140,75){\line(1,0){5}}
		\put(155,75){\line(1,0){10}}
		
		\put(15,90){\line(1,0){10}}
		\put(35,90){\line(1,0){5}}
		\put(80,90){\line(1,0){20}}
		\put(140,90){\line(1,0){5}}
		\put(155,90){\line(1,0){10}}
		
		\put(15,105){\line(1,0){10}}
		\put(35,105){\line(1,0){5}}
		\put(80,105){\line(1,0){20}}
		\put(140,105){\line(1,0){5}}
		\put(155,105){\line(1,0){10}}
		
		\put(25,10){\framebox(10,10){$X$}}
		\put(25,70){\framebox(10,10){$H$}}
		\put(25,85){\framebox(10,10){$H$}}
		\put(25,100){\framebox(10,10){$H$}}
		\put(145,70){\framebox(10,10){$H$}}
		\put(145,85){\framebox(10,10){$H$}}
		\put(145,100){\framebox(10,10){$H$}}
		
		\put(40,25){\framebox(40,85){$U$}}
		\put(100,25){\framebox(40,85){$U^\dagger$}}
		
		\put(85,10){\framebox(10,10){$Z$}}
		\put(90,20){\line(0,1){10}}
		
		\put(90,30){\circle*{3}}
		
	\end{picture}
	\caption{\small Quantum circuit $\cR^{U, U^\dagger}$.}
	\label{fig:ACfig1}
\end{figure}

The following immediate corollary follows from examining the algorithm in Figure~\ref{fig:ACfig1}.
\begin{corollary}\label{cor:qgl}
	In the setting of Theorem~\ref{thm:qgl}, if there exists a pure state $\ketphi$ such that
	\begin{align}\label{eq:absepsilon}
		\abs{ \Ex_r \left[ \norm{(I \otimes \bra{r \cdot x}) U_r\ketphi }^2 \right] -1/2}  \ge \epsilon~,
	\end{align}
	then it holds that 
	\begin{align}
		\norm{(\bra{x}\otimes I) \cR^{U, U^\dagger}\ket{0^{n}} \ketphi \ketz}^2 \ge (2\epsilon)^2~,
	\end{align}
	and therefore it is possible to recover the value of $x$ with probability at least $(2\epsilon)^2$ by measuring the first $n$ qubits in the computational basis.
	
	Furthermore, if $U_r$ acts as identity on a part of $\ketphi$, then $\cR$ does not require access to this part of its input. In other words, the above holds even in the presence of quantum auxiliary input.
	\end{corollary}
\begin{proof}
	Let us start by considering the case where 
	\begin{align}
		\Ex_r \left[ \norm{ (I \otimes \bra{r \cdot x}) U_r\ketphi }^2\right]   \ge  1/2 + \epsilon~.
	\end{align}
	Notice that in the corollary we only argue about the outcome of the measurement of the first $n$ qubits of the state in the computational basis, and not about the entire output state. For this purpose, replacing the $m$-qubit state $\ket{0^m}$ with $\ketphi$ does not matter since we can imagine a (possibly inefficient) unitary $W$ that generates $\ketphi$ from $\ket{0}$, then we can apply Theorem~\ref{thm:qgl} with a unitary $U'$ which first applies $W$ to create $\ketphi$ and then applies $U$. This would be identical to applying $U$ with the state $\ketphi$ as input. For applying $U'^\dagger$, we would need to apply $U^\dagger$ followed by $W^\dagger$, but the latter does not affect the outcome of the measurement of the first $n$ qubits and can therefore be omitted. We have that we can use the circuit in Figure~\ref{fig:ACfig1} with our $U$ and with $\ketphi$ as the $m$-qubit input. Furthermore, if $U$ itself acts as identity on some part of $\ketphi$, then $U^\dagger$ acts as identity on it as well, and it does not effect the measured outcome. We therefore get
	\begin{align*}
		\norm{(\bra{x}\otimes I) \cR^{U, U^\dagger}\ket{0^{n}} \ketphi \ketz}^2 &= \norm{(\bra{x}\otimes I) \cR^{U', U'^\dagger}\ket{0^{n+m+1}} }^2\\
		& \ge \abs{\bra{x, 0^m, 1} \cR^{U', U'^\dagger}\ket{0^{n+m+1}}}^2 \\
		& \ge (2\epsilon)^2~.
	\end{align*}
	
Finally, we can introduce an absolute value into Eq.~\eqref{eq:absepsilon} by considering the case where 
	\begin{align}
	\Ex_r \left[ \norm{ (I \otimes \bra{r \cdot x}) U_r\ketphi }^2\right]   \le 1/2 - \epsilon~.
	\end{align}
Here, can just consider $U'_r$ that always flips the last bit of the output. For $U'$ we have the ``correct'' expected value of at least $1/2 - \epsilon$, so we can apply Theorem~\ref{thm:qgl}. However, flipping the last bit in $U'$ is equivalent to conjugating the controlled-$Z$ gate in Figure~\ref{fig:ACfig1} by $X$ on the control qubit. This is equivalent to applying an additional $Z$ gate on the last qubit right after the controlled $Z$ operation, which again does not change the measured value of the first $n$ qubits. Therefore, applying the algorithm on $U$ would yield the same outcome as applying it on $U'$ and the corollary follows.
\end{proof}

\paragraph{An Estimation Lemma.} We use the following lemma. The proof is straightforward and is deferred to Appendix~\ref{apx:proofs}.

\begin{lemma}\label{lem:est}
		Let $n \in \bbN$ and let $(\Gamma_1, \ldots, \Gamma_n)$ be an arbitrary set of events over some probability space. Define $\Upsilon_i = \bigwedge_{j=1}^i \Gamma_j$ and let $\Upsilon_0$ be the universal event. Denote $\epsilon = \Pr[\Upsilon_n]$. Finally denote $\alpha_i = \Pr[\Gamma_i | \Upsilon_{i-1}]$.
	Then for all $\delta \in (0,1)$
	\begin{align}
		\Pr_{i \in [n]}\left[\alpha_i  > 1- \frac{\ln(1/\epsilon)}{\delta n}\right] \ge 1- \delta~.
	\end{align}
\end{lemma}
If particular for any value of $t$, setting $\delta=1/2$, $n > 2 t \ln(1/\epsilon)$, we get that there exists a value of $i$ (in fact, for half the values of $i$), $\alpha_i > 1-1/t$.

\subsection{EFI Quantum States}

The notion of efficiently generatable, statistically far but computationally indistinguishable pair of quantum states (EFI) was recently explicitly introduced in \cite{BCQ22}, extending ideas from \cite{Yan22}, as a generalization of a similar notion that existed in classical cryptography \cite{G90}.

\begin{definition}[EFI]\label{def:efi}
		An Efficiently Sampleable, Statistically Far and Computationally Indistinguishable pair of distributions (EFI) is an ensemble of equal-dimension pairs of  quantum states such that the following all hold:
		\begin{itemize}
			\item[(E)] Each element in the pair, as an individual ensemble, is efficiently generatable.
			\item[(F)] The two ensembles are not statistically indistinguishable.
			\item[(I)] The two ensembles are computationally indistinguishable.
		\end{itemize}
We also consider the notion of ``strong EFI'' for which condition (F) is strengthened as
\begin{itemize}
	\item[(F!)] The two ensembles have statistical distance $1-\eta$, where $\eta$ is a negligible function.
\end{itemize}

\noindent We may strengthen this notion even further and define ``perfect EFI'' via
\begin{itemize}
	\item[(F!!)] The two ensembles have statistical distance $1$.
\end{itemize}
\end{definition}

We note that by using repetition, strong EFI exists if and only if ``standard'' EFI exists (however, this equivalence is not known for perfect EFI).
\begin{proposition}\label{prop:efi}
	Strong EFI exists if and only if EFI exists.
\end{proposition}

We may now make the cryptographic assumption that EFI exist.
\begin{assumption}
	\label{a:efi}
EFI exist.
\end{assumption}

It has been shown in \cite{BCQ22}, also relying on prior work, that the existence of EFI is equivalent to that of a few cryptographic objects.
\begin{theorem}%
	EFI exists if and only if any of the following cryptographic primitives exist:\footnote{This refers to the standard adversarial cryptographic model, without making physical assumptions about limits on storage or use of special hardware.}
	\begin{enumerate}
		\item Quantum bit commitments.
		\item Quantum oblivious transfer protocols.
		\item Quantum secure multiparty computation protocols for non-trivial functionalities.
		\item Non-trivial quantum computational zero-knowledge proofs for the complexity class $\mathrm{QIP}=\mathrm{PSPACE}$.
	\end{enumerate}
\end{theorem}

\section{Black-Hole Radiation Decoding}

We provide our formal definition for the problem of decoding black-hole radiation. We define, from the beginning, the general version in which the decoder is allowed to declare failure. Essentially, we require that with non-negligible probability failure is not declared, and when this happens, the radiation decoder distills a qubit from the radiation that is highly entangled with a qubit inside the horizon. This is defined in terms of fidelity with the EPR state. A successful decoding will lead to non-trivial fidelity of more than $1/4$ (which is the fidelity of the EPR state with completely unentangled qubits). The formal definition follows.

\begin{definition}[Black-Hole Radiation Decoder]\label{def:rdecoder}
	Let $\ketpsi_{HBR}$ be a pure state supported over $3$ registers $HBR$, where $B$ is a single-qubit register. A \emph{radiation decoder} (or just ``decoder'' when the context is clear) is a quantum channel $\cD[R \to DF]$, where $D,F$ are each a single-qubit register. We refer to $F$ as the ``failure qubit'' and to $D$ as the ``decoded value''.
	
	For $\gamma, \epsilon \in [0,1]$ We say that $\cD$ is $(\gamma, \epsilon)$-decoder (against $\ketpsi$) if the following holds. Consider the experiment where $HBR$ are initialized to $\ketpsi_{HBR}$, and then $\cD[R \to DF]$ is applied. Then the reduced density matrix of $FBD$ is
	\begin{align}
		\rho_{FBD} = \gamma' \ketbra{0}_F \otimes (\rho_0)_{BD} + (1-\gamma') \ketbra{1}_F \otimes (\rho_1)_{BD}~,
	\end{align}
for some $\gamma' \ge \gamma$, and unit trace $\rho_0, \rho_1$, and in addition $\braepr\rho_0\epr \ge (1/4)+(3/4)\epsilon$.

We use the following terminology for asymptotic families of decoders:
\begin{itemize}
	\item A $(1,1-\eta)$-decoder with negligible $\eta$ is a \emph{strong decoder}.
	\item A $(\gamma,\epsilon)$-decoder which is computationally efficient and for which $\gamma \cdot \epsilon$ is non-negligible, is an \emph{effective decoder}. If in addition $\epsilon = 1-o(1)$, we call it a \emph{high confidence} decoder.
\end{itemize}
\end{definition}
The scaling for $\epsilon$ is chosen so that the effective range of $\epsilon$ is $[0,1]$ (this unfortunately leads to an asymmetry between $\epsilon >0$ and $\epsilon<0$ but it will not be problematic for us).

\begin{assumption}[Black-Hole Radiation Decoding ($\bhrd$)]\label{a:bhrd}
	There exists an ensemble of efficiently generated pure states as per Definition~\ref{def:rdecoder}, for which there exists a strong (inefficient) decoder, but there does not exist an effective decoder.
\end{assumption}

The following is a simple observation.
\begin{proposition}
If there exists a $(\gamma, \epsilon)$-decoder against a state $\ketpsi$ then there also exists a $(1, \gamma \cdot \epsilon)$-decoder against $\ketpsi$.
\end{proposition}
\begin{proof}
Given a $(\gamma, \epsilon)$-decoder $\rdec'$, consider a decoder $\rdec[R \to DF]$ as follows that first runs the original $\rdec'[R \to D'F']$ and measures the register $F'$. If $f'=0$ then set $D=D'$, otherwise set $D$ to contain a maximally mixed state. Set $F$ to $\ketbra{0}$. The result follows by calculation.
\end{proof}

\subsection{The Superdense Decoding Problem}
\label{sec:superdense}

\def\esdc{\mathsf{SupDec}}

We introduce a computational problem that is related to the well known notion of quantum superdense coding. We relate it to the problem of radiation decoding and show that the two are essentially equivalent. We believe that the superdense formulation of the problem may be a useful one, and indeed this formulation plays an important role in our constructions of $\efi$ in Section~\ref{sec:bhrdtoefi}.

The formal definition follows. We wish to emphasize that both the radiation decoder and the superdense decoder are defined with respect to the a state (or an ensemble of states) over $HBR$, but whereas a radiation decoder only takes $R$ as input, a superdense decoder takes $BR$ as input (where a classical challenge is encoded into the register $B$).

\begin{definition}\label{def:supdecoder}
	Let $\ketpsi_{HBR}$ be a pure state supported over $3$ registers $HBR$, where $B$ is a single-qubit register. A Decoder is a quantum channel $\cS[BR \to PF]$, where $P$ is a two-qubit register and $F$ is a single-qubit register. We consider the following experiment:

	\begin{enumerate}
		\item Initialize the registers $HBR$ to $\ketpsi_{HBR}$.
		\item Sample two random bits $x,z$.
		\item Apply the Pauli mask $X^x Z^z$ to the register $B$ (this is also known as a quantum one-time-pad encryption).
		\item Apply $\sdec[BR \to PF]$.
		\item Measure $P, F$ in the computational basis to obtain values $(x',z'), f$.
	\end{enumerate}

	For $\gamma, \epsilon \in [0,1]$, we say that $\sdec$ is $(\gamma, \epsilon)$-superdense-decoder (against $\ketpsi$) if $\Pr[f = 0] \ge \gamma$ and $\Pr[(x',z') = (x,z) | f=0] \ge 1/4+(3/4)\epsilon$.
\end{definition}

The following lemma establishes the connection between superdense decoding and radiation decoding (in any parameter regime).
\begin{lemma}\label{lem:rdecsdec}
	Let $\ketpsi_{HBR}$ be a pure state supported over $3$ registers $HBR$, where $B$ is a single-qubit register. Then there exists a $(\gamma, \epsilon)$-radiation decoder against $\ketpsi$ if and only if there exists a $(\gamma, \epsilon)$-superdense decoder against $\ketpsi$ with comparable computational complexity.
\end{lemma}
\begin{proof}
	Let $\ketpsi_{HBR}$ be as in the lemma statement. We consider both directions of the lemma.
	
	\paragraph{Radiation-Decoder Implies Superdense-Decoder.} This direction is almost immediate. Given a radiation decoder $\rdec$, we can use it to decode the radiation $R$. If the experiment is successful, then whenever in the original experiment $BD$ constitute an EPR pair, we can use standard superdense decoding procedure to recover $x,z$.
	
	More formally, let $\rdec$ be a $(\gamma,\epsilon)$-radiation decoder. We define a superdense decoder $\sdec[BR \to PF]$ as follows.
	\begin{enumerate}
		\item Apply $\rdec[R \to DF]$.
		\item Measure $BD$ in the Bell basis and store the outcome in $P$.
	\end{enumerate}
	
	The parameters $(\gamma,\epsilon)$ remain the same by the properties of superdense coding.

	\paragraph{Superdense-Decoder Implies Radiation-Decoder.} Let $\sdec$ be a $(\gamma, \epsilon)$-superdense decoder.
	We start by noticing that applying a quantum one-time pad to a register is equivalent to performing quantum teleportation onto a different variable. Thus, the experiment from Definition~\ref{def:supdecoder} is equivalent to the following.
	\begin{enumerate}
		\item Initialize the registers $HBR$ to $\ketpsi_{HBR}$.
		\item Initialize two single-qubit registers $DT$  to $\epr_{DT}$.
		\item \label{s:BTepr} Measure $BD$ in the Bell basis to obtain values $x,z$. (This is the process of teleporting the register $B$ into the register $T$ via the register $D$.)
		\item Apply $\sdec[TR \to PF]$.
		\item Measure $P$ in the computational basis to obtain values $(x',z')$.
	\end{enumerate}	
	
	We therefore conclude that in the above experiment, $\sdec$ outputs $F$ which takes the value $0$ with probability at least $\gamma$, and letting $f$ being the measured value of $F$, we have $\Pr[(x',z') = (x,z) | f=0] \ge 1/4+(3/4)\epsilon$. 	
	
	We note that step~\ref{s:BTepr} in the above experiment (the $BD$ measurement) commutes with all following steps of the experiment until the final one. Let us therefore defer it until the end of the experiment.

	We now define a radiation decoder $\rdec[R \to DF]$ as follows.
\begin{enumerate}
	\item Initialize two single-qubit registers $DT$  to $\epr_{DT}$.
	\item Apply $\sdec[TR \to PF]$.
	\item Measure $P$ in the computational basis to obtain values $x',z'$.
	\item Apply $Z^{z'}X^{x'}$ to $D$.
	\item Output $DF$.
\end{enumerate}	
	
	Measuring the value of $F$ in the computational basis to obtain a value $f$, we have that $\Pr[f=0] \ge \gamma$ by the properties of $\sdec$. Now let us consider the reduced density matrix of $PBD$ conditioned on $f=0$ (this is a trace-$1$ density matrix). This matrix can be written as
	\begin{align}
		\rho_{PBD} = \sum_{x',z'} c_{x',z'} \ketbra{x',z'}_{P} \otimes (\rho_{x',z'})_{BD}~,
	\end{align}
	where $c_{x',z'}$ are real values summing to $1$, and $\rho_{x',z'}$ are proper trace-$1$ density matrices. We know that the probability of measuring $BD$ in the Bell basis and obtaining values that are equal to $(x',z')$ is $1/4+(3/4)\epsilon$. We notice that the success probability would have been the same if we performed $X^{x'}Z^{z'}$ on $D$ and applied the $\kbepr$ projector, which is exactly the success probability of $\rdec$, conditioned on $f=0$. This concludes this case. %
\end{proof}

\section{Cryptology-Induced Gravitational Phenomena}

We show that Assumption~\ref{a:efi} implies Assumption~\ref{a:bhrd}. An efficiently generatable state $\ketpsi_{HBR}$ for which it is hard to decode entanglement. We in fact show that the existence of $\efi$ also implies such states where it is impossible to detect \emph{any} correlation between the $B$ and $R$ registers. In other words, $BR$ are computationally indistinguishable from being in tensor product. This is therefore a stronger result than that of \cite{HH13,A16}, where some classical or Hadamard correlations can be efficiently recovered.

\begin{theorem}
	If there exists an EFI pair as per Assumption~\ref{a:efi} then there exists an ensemble of hard-to-decode radiation states as per Assumption~\ref{a:bhrd}.
\end{theorem}
\begin{proof}
	Consider a strong EFI pair of states (recalling Proposition~\ref{prop:efi}, EFI implies strong EFI). By considering the purification of the generating channel of the pair of states, there exist efficiently generatable ensembles of \emph{pure} states $\{\ket{\varphi_b}\}_{b \in \binset}$ over registers $HR'$ such that $\tr_H\left[ \ketbra{\varphi_b}_{HR'} \right]$ are strong EFI.
	
	Consider the ensemble of states defined over registers $H'BTH_1R'_1H_2R'_2$ as
	\begin{align}
		\ketpsi &= \frac{1}{2\sqrt{2}}\sum_{b,x,z \in \binset} \ket{xz}_{H'} \ket{b}_B (X^x Z^z\ket{b}_T) \ket{\varphi_x}_{H_1R'_1} \ket{\varphi_z}_{H_2R'_2}~.
	\end{align}
	We denote $H = H' H_1 H_2$, $R = T R'$ so that $\ketpsi$ is defined over $HBR$.

	\paragraph{Efficient Generation.} The above state is generated in the following steps:
	\begin{enumerate}
		\item Generate an EPR pair in $BT$.
		\item Generate $\ket{++}$ in $H'$.
		\item Apply a ``quantum one-time pad'' $X^x Z^z$, controlled by $H'$ on the qubit in $T$. (This is equivalent to converting $BT$ to $\ket{\eprtxt_{xz}}$, controlled by the values $xz$ in $H'$.)
		\item Encode $x,z$ using the $\efi$. Formally, controlled by the individual bits of $H'$, generate either $\ket{\varphi_0}$ or $\ket{\varphi_1}$ in $H'_1 R'_1$ controlled by the first qubit of $H'$, and likewise in $H'_2 R'_2$ controlled by the second qubit of $H'$.
	\end{enumerate}
	
	\paragraph{Inefficient Decoder.}
	By the strong EFI property, there exists a possibly inefficient channel $\distn$ so that $\distn[R' \to D]$ distinguishes $\ketbra{\varphi_1}$ from $\ketbra{\varphi_0}$ with $(1-\eta)$ advantage, for a negligible $\eta$. 
	
	We now show an inefficient $(1, 1-(4/3)\eta)$-radiation decoder $\rdec[R \to TF]$ (note that the register $T$ will be our output register).
	\begin{enumerate}
		\item Apply $\cD[R'_1 \to D_1]$.
		\item Apply $\cD[R'_2 \to D_2]$.
		\item Apply $\cz_{D_2 T}$.
		\item Apply $\cnot_{D_1 T}$.
		\item Initialize $F$ to $\rho_F = \ketbra{0}$.
	\end{enumerate}
	
	To analyze, let us consider the state of $H'BTD_1D_2$ after the two applications of $\cD$.
	Consider an experiment where the $\cz$ and $\cnot$ gates are controlled by $H'_2, H'_1$ instead of $D_2, D_1$. In that case, essentially by definition, $BT$ is exactly in $\epr$. Therefore, the probability that $BT$ is not measured in $\epr$ in the experiment is at most the probability that $H'$ is not equal to $D_1 D_2$, which is at most $\eta$ by Proposition~\ref{prop:distinguish} and the union bound: $2\cdot (1- (1/2+(1-\eta)/2))=\eta$.

	\paragraph{No Effective Decoder.}
	Assume there exists an efficient $(\gamma, \epsilon)$-radiation-decoder $\rdec$ for the state $\ketpsi$. We show how it can be used to construct an efficient distinguisher $\distn$ for the EFI with advantage $\gamma \cdot \epsilon$. This would show that an effective decoder contradicts the computational indistinguishability property of the EFI.
	
	We will use a hybrid argument. We describe a sequence of experiments (hybrids) that will allow us to draw the required conclusion.
	
	\begin{itemize}
		\item Hybrid $0$. This is the standard radiation decoding game as per Definition~\ref{def:rdecoder}, where $\rdec$ attempts to decode the radiation in encoded in $\ketpsi$.
		
		\item Hybrid $1$. In this experiment, we change the way $\ketpsi$ is generated. In particular we generate $H_2 R'_2$ to always contain $\ket{\varphi_0}$, regardless of the value of $x$.
		
		\item Hybrid $2$. In this experiment, we again change the way $\ketpsi$ is generated. We keep $H_2 R'_2$ as in the previous hybrid, but now we also generate $H_1 R'_1$ to always contain $\ket{\varphi_0}$, regardless of the value of $z$.
	\end{itemize}

We let $\gamma_i$ denote the probability of measuring $0$ in $F$ in Hybrid~$i$, and by $\epsilon_i$ the probability of measuring EPR in $BD$ conditioned on measuring $0$ in $F$. We argue that it must be the case that 
$\abs{\gamma_i\epsilon_i - \gamma_{i+1}\epsilon_{i+1}}$,
are negligible. Otherwise, by Proposition~\ref{prop:indisthybrid}, it would allow us to distinguish the EFI states. %

Finally, it must be the case that in Hybrid~$2$, even when $F$ measures to $0$, the probability to measure $\epr$ in $BD$, no matter how $D$ is generated, is exactly $1/4$. 	In fact, in Hybrid~$2$, it is indeed the case that $R$ and $B$ are in tensor product, since $R'_1, R'_2$ are independent of $x,z$, and therefore $T$ is in maximally mixed state. Indeed, in this hybrid there simply does not exist \emph{any} correlation between $B$ and $R$. Since the reduced density matrix of $B$ is maximally mixed, we have that the probability of measuring EPR is exactly $1/4$.  The theorem thus follows.
\end{proof}

\section{Gravitationally-Induced Cryptology}
\label{sec:bhrdtoefi}

We start by showing the ``vanilla'' version of our claim. We consider the following two distributions.

\begin{definition}\label{def:efiplain}
	Let $\ketpsi$ be an efficiently generatable ensemble of states over registers $HBR$. We consider the following EFI candidate distribution. Specifically, the distribution $\rho_{b}$ for $b \in \binset$ is defined as follows.
	\begin{enumerate}
		\item Initialize $HBR$ to $\ketpsi_{HBR}$.
		
		\item Sample uniformly $x,z\in \binset$.
		
		\item Apply $(X^{x}Z^{z})_{B}$.
		
		\item If $b=0$ set $x'=x$, $z' =z$, otherwise (if $b=1$) sample $x', z' \in \binset$ uniformly at random.
		
		\item Output $BR$ and $x'z'$.
	\end{enumerate}
\end{definition}

\begin{theorem}
Let $\ketpsi$ be efficiently generatable, strongly decodable but not effectively decodable, as per Assumption~\ref{a:bhrd}. Then the pair of distributions in Definition~\ref{def:efiplain} is EFI as per Assumption~\ref{a:efi}.
\end{theorem}
\begin{proof}
Since $\ketpsi$ is strongly radiation-decodable, it is also strongly superdense decodable. We show an inefficient distinguisher for the $\efi$ candidate as follows. Given as input $BRx'z'$, apply the inefficient $(1,1-\eta)$-superdense decoder to $BR$ to obtain values $x'' z''$. If these values are equal to $x'z'$ output $0$, otherwise output $1$. The probability that this distinguisher outputs $1$ on $\rho_0$ is at most $(4/3)\eta$. The probability that it outputs $1$ on $\rho_1$ is at least $(3/4)(1-(4/3)\eta)$. The distinguishing gap thus follows.	
	
On the other hand, let $\distn$ be a distinguisher between $\rho_0, \rho_1$ as per Definition~\ref{def:efiplain}, with advantage $\delta$. We show that there exists a $(1,\delta/12)$-superdense decoder $\sdec$ against $\ketpsi$ with roughly the same complexity. Thus, by Lemma~\ref{lem:rdecsdec}, there is also a radiation decoder against $\ketpsi$ and the theorem would follow.

Our superdense decoder, upon receiving an input in registers $BR$, does the following.
\begin{enumerate}
	\item Sample $x',z' \in \binset$ uniformly at random.
	
	\item Apply $\distn[BRx'z' \to T]$ and measures the outcome in the computational basis to obtain a value $t$.
	
	\item If $t=0$, set $P$ to the classical value $xz = x'z'$, otherwise set $P$ to a maximally mixed state.
	Set $F$ to $\ketbra{0}$.
\end{enumerate}

Note that by definition, this $\sdec$ has $\gamma=1$. We show that its $\epsilon$ is related to the distinguishing advantage $\delta$. To this end, we let $\delta_b$ denote the probability that $\distn$ outputs $0$ on $\rho_b$. We denote by $\delta_{\bot}$ the probability that $\distn$ outputs $0$ when its input consists of properly generated $BR$, along with $x'z'$ that are sampled to be \emph{not equal} to the actual $xz$. A simple calculation implies that $\delta_1 = \tfrac{1}{4}\delta_0 + \tfrac{3}{4}\delta_\bot$. By the definition of $\distn$ we have that $\delta_0 - \delta_1 = \delta$.

The probability that $\sdec$ outputs the correct values $x,z$ in the superdense decoding experiment is therefore:
\begin{align}
	\frac{1}{4} \delta_1 \cdot 1 + \frac{1}{4} (1-\delta_1) \cdot \frac{1}{4} + \frac{3}{4}(1-\delta_\bot)\cdot \frac{1}{4}~.
\end{align}
The three terms in this expression are as follows:
\begin{enumerate}
	\item If $x'z'$ were sampled identically to the real $xz$, which happens with probability $1/4$ then:
	\begin{enumerate}
		\item With probability $\delta_1$, $\distn$ outputs $0$ and we output the correct value. This event therefore has a total probability of $\frac{1}{4} \delta_1 \cdot 1$.
		
		\item With probability $(1-\delta_1)$, $\distn$ outputs $1$ and we output a random value, which hits the correct value with probability $1/4$. This event therefore has a total probability of $\frac{1}{4} (1-\delta_1) \cdot \frac{1}{4}$.	
	\end{enumerate}

\item If $x'z'$ were sampled to not be equal to $xz$, which happens with probability $3/4$, then we can only win if $\distn$ outputs $1$, which happens with probability $(1-\delta_\bot)$, and in addition the randomly sampled output hits the correct values $xz$, which happens with probability $1/4$. This event therefore has a total probability of $\frac{3}{4}(1-\delta_\bot)\cdot \frac{1}{4}$.
\end{enumerate}

The $3$ terms above sum to $\frac{1}{4} + \frac{\delta}{16}$. It follows that indeed $\epsilon = (4/3)\cdot (\delta/16)$ and the result follows. %
\end{proof}

\subsection{Amplifying the Result}

We now show a candidate for $\efi$ which is proven to remain $\efi$ even if the underlying state is only mildly hard to decode.

\begin{definition}\label{def:efifrombhrd}
Let $\ketpsi$ be an ensemble of states, and let $n=n(\secp)$ be a parameter. We consider the following EFI candidate distribution. Specifically, the distributions $\rho_0, \rho_1$ that are defined as follows.
\begin{enumerate}
	\item Consider registers $H_i B_i R_i$, for $i \in [n]$, initialize each triple $H_i B_i R_i$ to $\ketpsi_{H_i B_i R_i}$.
	
	\item Sample uniformly random strings $\vec{x}, \vec{z} \in \binset^n$.
	
	\item For all $i$, apply $(X^{x_i}Z^{z_i})_{B_i}$.
	
	\item Sample uniformly random strings $\vec{a}, \vec{a}' \in \binset^n$.
	
	\item Calculate $c = (\vec{a} \cdot \vec{x}) \oplus (\vec{a}' \cdot \vec{z})$, where the inner product is over the binary field.
	
	\item The output of $\rho_b$ contains the following: 
	\begin{align}
		\left(\{ B_i R_i \}_{i\in [n]}, (\vec{a}, \vec{a}'), c' = c \oplus b\right)~.
	\end{align}
\end{enumerate}
\end{definition}

We now prove that it suffices that $\ketpsi$ is mildly hard to decode in order to imply that the above construction is $\efi$. For the sake of convenience, we split our argument into two lemmas. One that considers the efficiency of generation, and the existence of an inefficient distinguisher, and the other that converts an efficient distinguisher into an efficient high-confidence decoder.

\begin{lemma}
	Let $\ketpsi$ be efficiently generatable and strongly (possibly inefficiently) decodable as per Assumption~\ref{a:bhrd}. Then the pair of distributions in Definition~\ref{def:efifrombhrd} is efficiently generatable and statistically distinguishable for any polynomial $n(\secp)$.
\end{lemma}
\begin{proof}
Since $\ketpsi$ is efficiently generatable, then our $\rho_b$ are also efficiently generatable for every polynomial $n$, since the generation of $\rho_b$ contains a generation of polynomially many copies of $\ketpsi$ as well as a polynomial number of elementary comptational steps.

Next we show that $\rho_0, \rho_1$ are inefficiently distinguishable, given that there exists a strong radiation decoder $\rdec$ against $\ketpsi$. We use Lemma~\ref{lem:rdecsdec} to argue that there also exists a strong $(1,1-\eta)$-superdense decoder $\sdec$ against $\ketpsi$ and use $\sdec$ to derive a distinguisher between $\rho_0, \rho_1$ with advantage at least $1-2n\eta$. Since $\eta$ is negligible and $n$ is polynomial, then $2 n \eta$ is also negligible. To see the above, we consider a distinguisher $\distn$ that takes as input a register $S$ containing $\left(\{ B_i R_i \}_{i\in [n]}, (\vec{a}, \vec{a}'), c'\right)$ and $\distn[S \to T]$ acts as follows:
\begin{enumerate}
	\item Apply $\sdec[B_i R_i \to P_i F_i]$, for all $i$. 
	\item Measure $P_i$ in the computational basis to obtain two classical bits $(x'_i, z'_i)$.
	\item Setting $\vec{x}' = (x_1, \ldots, x_k)$, $\vec{z}' = (z_1, \ldots, z_k)$, compute $b' = c' \oplus (\vec{a} \cdot \vec{x}') \oplus (\vec{a}' \cdot \vec{z}')$.
	\item Initialize a register $T$ to the classical value $b'$ and produce this register as the output. 
\end{enumerate}
We note that since $\cB$ has $\gamma=1$, we need not measure the register $F_i$ since it is guaranteed to always be identical to $0$.

To see the distinguishing gap of $\distn$ defined above, we notice that since $\sdec$ is a strong decoder and therefore with all but $\eta$ probability, $(x'_i, z'_i)=(x_i, z_i)$. Applying the union bound, we get that this holds for all $i$ with all but $n \eta$ probability. If this indeed holds then, by definition, $b'=b$. The distinguishing advantage thus follows.
\end{proof}

The next lemma establishes that an efficient distinguisher against our $\efi$ candidate would result in a high-confidence decoder, so long as $n$ is chosen to be sufficiently large, specifically $n(\secp) = \omega(\log \secp)$.
\begin{lemma}\label{lem:disttodec}
	Let $\ketpsi$ be efficiently generatable, and consider the pair of distributions in Definition~\ref{def:efifrombhrd}.
	If the aforementioned distributions are efficiently distinguishable with non-negligible advantage, then there is a $(\gamma,\epsilon)$-decoder with polynomially-related computational complexity against $\ketpsi$, where $\gamma$ is non-negligible, and $\epsilon = 1-O(\log(\secp)/n)$.
\end{lemma}

\begin{proof}
For the remainder of the proof, let $\ketpsi$ be an ensemble of states and let $\distn$ be an efficient distinguisher between the pair of distributions of Definition~\ref{def:efifrombhrd}, with non-negligible advantage $\delta$. We show that there exists an effective radiation decoder $\rdec$ against $\ketpsi$. We will do this by first presenting a superdense decoder $\sdec$ against $\ketpsi$ and then applying Lemma~\ref{lem:rdecsdec} to derive the final conclusion.

\paragraph{From Distinguishing to Batch-Decoding.} We start by showing that the existence of $\distn$ also implies the existence of an efficient \emph{batch-decoder} $\cB$ that succeeds with probability at least $\delta' = \poly(\delta)$ in the following experiment.
\begin{enumerate}
	\item Consider registers $H_i B_i R_i$, for $i \in [n]$, initialize each triple $H_i B_i R_i$ to $\ketpsi_{H_i B_i R_i}$.

	\item Sample uniformly random strings $\vec{x}, \vec{z} \in \binset^n$.
	
	\item For all $i$, apply $(X^{x_i}Z^{z_i})_{B_i}$.
	
	\item Apply $\cB[\{ B_i R_i \}_{i\in [n]} \to Q]$, where $Q$ is a $2n$-qubit register.
	
	\item Measure $Q$ in the computational basis to obtain values $\vec{x}', \vec{z}' \in \binset^n$.
	
	\item The experiment succeeds if $(\vec{x}', \vec{z}') = (\vec{x}, \vec{z})$.
\end{enumerate}

In other words, with probability $\delta'$, $\cB$ is able to recover all $x_i, z_i$ values for all $i$, given only access to the registers $B_i R_i$.

The existence of $\cB$ is implied from that of $\distn$ by Corollary~\ref{cor:qgl} (the quantum Goldreich-Levin theorem). To apply the corollary, we need to construct the required unitary $U$ from the distinguisher $\distn$, using a standard argument. 

Let $\delta_{\vec{x},\vec{z}}$ denote the advantage of $\distn$ over fixed values of $\vec{x},\vec{z}$. With probability at least $\delta/2$ over $\vec{x},\vec{z}$, it holds that $\abs{\delta_{\vec{x},\vec{z}}} \ge \delta/2$.
Consider the channel $\cD'$ that, given as input $\left(\{ B_i R_i \}_{i\in [n]}, (\vec{a}, \vec{a}')\right)$ (note that $c'$ is not given) samples a random value $t' \in \binset$, applies $\distn$ on $\left(\{ B_i R_i \}_{i\in [n]}, (\vec{a}, \vec{a}'), t'\right)$ to obtain a value $t$, and outputs $t' \oplus t$. It is not hard to verify that the probability that $t'\oplus t = c = (\vec{a} \cdot \vec{x}) \oplus (\vec{a}' \cdot \vec{z})$ is exactly equal to $1/2+\delta_{\vec{x},\vec{z}}/2$. The unitary $U$ is derived from the purification of the channel $\cD'$, in addition to acting as identity on the registers $H_i$ which are not provided to $\cD$. The state $\ketphi$ in the corollary is the pure state in $\{ H_i B_i R_i \}_{i\in [n]}$ conditioned on the values of $\vec{x}, \vec{z}$, which are constants for the purpose of the corollary.

It follows that with probability at least $\delta/2$ over $\vec{x},\vec{z}$, the correct values are recovered with probability at least $(\delta/2)^2$, which implies that $\cB$ succeeds in decoding $\vec{x},\vec{z}$ with probability at least $\delta' = (\delta/2)^3$.

Corollary~\ref{cor:qgl} can therefore be applied and the performance of $\cB$ follows.

\paragraph{From Batch-Decoding to Single-Instance Decoding.} We now use the batch decoder $\cB$ to obtain a (regular) effective superdense decoder $\sdec$ as follows.

We consider the following random variables, with reference to the experiment of $\cB$ above. We denote by $\Gamma_i$ the event where $(x'_i, z'_i) = (x_i, z_i)$, and note that all $\Gamma_i$ are defined over the same probability space. We can now invoke the estimation lemma (Lemma~\ref{lem:est}) to deduce that for at least half of the values $i \in [n]$ it holds that 
\begin{align}
	\alpha_i = \Pr\left[\Gamma_i \Big| \bigwedge_{j=1}^{i-1} \Gamma_i\right] > 1 - \frac{2\ln(1/\delta')}{n}~.
\end{align}
Since we know that $\delta'=\poly(\secp)$, it holds that $\frac{2\ln(1/\delta')}{n} = O(\log(\secp) / n)$.

We also recall that for all $i$ it holds that
\begin{align}
	\Pr\left[\bigwedge_{j=1}^{i-1} \Gamma_i\right] \ge \Pr\left[\bigwedge_{j=1}^{n} \Gamma_i\right] \ge \delta'~.
\end{align}

It therefore follows that in $\poly(1/\delta, n, \secp) = \poly(\secp)$ time, it is possible to estimate all values $\alpha_i$ to within additive error $O(\log \secp / n )$, with global estimation error $O(\log \secp / n )$, using the Chernoff bound. This is because the batch-decoding experiment can be ran efficiently (in polynomial time).

We therefore devise our superdense decoder $\sdec[BR \to PF]$ as follows.
\begin{enumerate}
	\item Using $\cB$, estimate the values of all $\alpha_i$ to within $O(\log \secp / n )$ with $O(\log \secp / n )$ total estimation error, as described above.
	\item Let $i^{*}$ be such that the estimated  $\alpha_{i^*}$ is the highest.
	\item For all $i \neq i^*$ generate $B_i R_i$ along with $x_i, z_i$ as in the batch-decoding experiment.
	\item Set $B_{i^*} R_{i^*} = BR$.
	\item Run $\cB[\{ B_i R_i \}_{i\in [n]} \to Q]$. Measure $Q$ in the computational basis to obtain $\vec{x}', \vec{z}' \in \binset^n$.
	\item If for all $i < i^*$ it holds that $(x'_i, z'_i) = (x_i, z_i)$, set $F$ to the classical value $0$ and set $P$ to the classical values $x'_i, z'_i$.
	\item Otherwise set $F$ to the classical value $1$ and set $P$ to an arbitrary value.
\end{enumerate}

It follows from the above argument that $\sdec$ is a $(\gamma, \epsilon)$ superdense decoder against $\ketpsi$ with $\gamma \ge \delta'$ and $\epsilon \ge (1 - O(\log \secp / n))$. %

Finally, applying Lemma~\ref{lem:rdecsdec}, we deduce that we can also obtain a $(\gamma, \epsilon)$-radiation decoder with the same parameters, which completes the proof of our lemma. %
\end{proof}

\appendix

\section{Proof of The Estimation Lemma}
\label{apx:proofs}

\begin{lemma}[Lemma ~\ref{lem:est}, Restated]
	Let $n \in \bbN$ and let $(\Gamma_1, \ldots, \Gamma_n)$ be an arbitrary set of events over some probability space. Define $\Upsilon_i = \bigwedge_{j=1}^i \Gamma_j$ and let $\Upsilon_0$ be the universal event. Denote $\epsilon = \Pr[\Upsilon_n]$. Finally denote $\alpha_i = \Pr[\Gamma_i | \Upsilon_{i-1}]$.
	Then for all $\delta \in (0,1)$
	\begin{align}
		\Pr_{i \in [n]}\left[\alpha_i  > 1- \frac{\ln(1/\epsilon)}{\delta n}\right] \ge 1- \delta~.
	\end{align}
\end{lemma}

\begin{proof}
	By definition it holds that
	\begin{align}
		\alpha_i = \frac{\Pr[\Gamma_i \land \Upsilon_{i-1}]}{\Pr[\Upsilon_{i-1}]} = \frac{\Pr[\Upsilon_{i}]}{\Pr[\Upsilon_{i-1}]}~.
	\end{align}
	
	Then by a telescopic product it holds that
	\begin{align}
		\prod_{i=1}^n \alpha_i = \Pr[\Upsilon_n] = \epsilon~,
	\end{align}
	or alternatively
	\begin{align}
		\frac{\ln(1/\epsilon)}{n} = \Ex_{i \in [n]}\left[ \ln(1/\alpha_i)\right]~.
	\end{align}
	We can apply Markov's inequality (since $\ln(1/\alpha_i)$ is positive) to conclude that for all $\delta \in (0,1)$:
	\begin{align}
		\Pr_i \left[\ln(1/\alpha_i) \ge \frac{\ln(1/\epsilon)}{\delta n} \right] \le \delta~.
	\end{align}
	
	Therefore with probability at least $1-\delta$ over $i$, we have that $\ln(1/\alpha_i) < \frac{\ln(1/\epsilon)}{\delta n}$, which implies that $1/\alpha_i < e^{\frac{\ln(1/\epsilon)}{\delta n}}$. That is
	\begin{align}
		\alpha_i > e^{-\frac{\ln(1/\epsilon)}{\delta n}} \ge 1 - \frac{\ln(1/\epsilon)}{\delta n}~.
	\end{align}
	This concludes the proof of the lemma.
\end{proof}

\ifnum\anon=0

\subsection*{Acknowledgements}

We thank Scott Aaronson, Ran Canetti, Isaac Kim and Luowen Qian for valuable feedback. We also thank anonymous reviewers for their comments.

\fi

%
%

\end{document}

%% file: style.tex
\newcommand{\remove}[1]{}
\newcommand{\ignore}[1]{}
\ifnum\llncs=0
\usepackage{amsthm}
\else
\fi

\ifnum\draft=1
\def\ShowAuthNotes{1}
\else
\def\ShowAuthNotes{0}
\fi

\usepackage{amsmath,amssymb}
\usepackage{url}
\usepackage{cite}

\usepackage[dvipsnames]{xcolor}
\definecolor{DarkBlue}{RGB}{0,0,150}
\definecolor{llg}{gray}{0.95}
\definecolor{lg}{gray}{0.85}

\usepackage[colorlinks,linkcolor=DarkBlue,citecolor=DarkBlue]{hyperref}

\ifnum\llncs=0

\newtheorem{theorem}{Theorem}[section]
\newtheorem{proposition}[theorem]{Proposition}
\newtheorem{definition}[theorem]{Definition}

\newtheorem{lemma}[theorem]{Lemma}

\newtheorem{assumption}{Assumption}

\newtheorem{corollary}[theorem]{Corollary}

\else

\spnewtheorem{prop}{Property}{\bfseries}{\itshape}
\spnewtheorem{fact}{Fact}{\bfseries}{\itshape}
\spnewtheorem{subclaim}{Claim}[theorem]{\bfseries}{\itshape}
\spnewtheorem{tlclaim}[theorem]{Claim}{\bfseries}{\itshape}
\spnewtheorem{assumption}{Assumption}{\bfseries}{\itshape}
\fi

\newcommand{\secparam}{\lambda}
\newcommand{\secp}{\secparam}

\ifnum\llncs=1
\def\pfend{\hfill\qedsymbol}
\else
\def\pfend{}
\fi

\def\cA{{\cal A}}
\def\cB{{\cal B}}
\def\cC{{\cal C}}
\def\cD{{\cal D}}

\def\cR{{\cal R}}
\def\cS{{\cal S}}
\def\cT{{\cal T}}

\def\bbC{{\mathbb C}}
\def\bbE{{\mathbb E}}

\def\bbN{{\mathbb N}}

\def\binset{\{0,1\}}

\def\q2{\lfloor q/2 \rceil}

\newcommand{\abs}[1]{\left\vert {#1} \right\vert}
\newcommand{\norm}[1]{\left\| {#1} \right\|}

\def\poly{{\rm poly}}

\def\bqp{\mathbf{BQP}}
\def\qma{\mathbf{QMA}}
\def\szk{\mathbf{SZK}}
\def\qip{\mathbf{QIP}}
\def\pspace{\mathbf{PSPACE}}
\def\ccp{\mathbf{P}}

\newcommand{\mx}[1]{\mathbf{{#1}}}

\newcommand{\Ex}{\mathop{\bbE}}

\newcommand{\boldpar}[1]{\vspace{3pt}\par\noindent\textbf{#1}}

\ifnum\ShowAuthNotes=1
\newcommand{\authnote}[3]{\textcolor{#3}{[{\footnotesize {\bf #1:} { {#2}}}]}}
\else
\newcommand{\authnote}[3]{}
\fi

\newcommand{\absnewline}{\ifnum\llncs=1 \\ \fi}

\def\mgp[#1]{\mx{G}_{#1}}

\def\mgip[#1]{\mx{G}_{#1}^{-1}}

\def\mcit[#1]{\mci[#1]^T}
\def\mci[#1]{\mx{C}_{#1}}

\newcounter{hybridcount}
\newcounter{prevhybridcount}
\newcounter{nexthybridcount}

\def\beginM{\left[\begin{matrix}}
\def\endM{\end{matrix}\right]}

\def\m0{\mx{0}}

\newcommand{\ket}[1]{|{#1}\rangle}
\newcommand{\bra}[1]{\langle{#1}|}
\newcommand{\ketbra}[1]{\ket{{#1}}\bra{{#1}}}

\def\ketz{\ket{0}}
\def\keto{\ket{1}}
\def\ketp{\ket{+}}
\def\ketm{\ket{-}}

\def\ketphi{\ket{\varphi}}

\def\ketpsi{\ket{\psi}}

\newcommand{\tr}{\mathop{\mathrm{Tr}}}

\def\eprtxt{\mathrm{epr}}
\def\epr{\ket{\mathrm{epr}}}
\def\braepr{\bra{\mathrm{epr}}}
\def\kbepr{\ketbra{\mathrm{epr}}}
\def\cnot{\mathrm{CNOT}}
\def\cz{\mathrm{CZ}}